\documentclass[12pt,twocolumn]{IEEEtran}

\IEEEoverridecommandlockouts                              

\usepackage{graphicx}

\newtheorem{corollary}{Corollary}[section]
\newtheorem{lemma}{Lemma}[section]

\def\theequation{\thesection.\arabic{equation}}

\newcommand{\be}{\begin{equation}}
\newcommand{\ee}{\end{equation}}
\newcommand{\bd}{\begin{displaymath}}
\newcommand{\ed}{\end{displaymath}}
\newcommand{\bea}{\begin{eqnarray}}
\newcommand{\eea}{\end{eqnarray}}

\newcommand{\Hi} {{\cal H}^\infty}

\title{\LARGE \bf Remarks on Strong Stabilization and Stable $\Hi$ Controller Design$^{*}$\thanks{$^*$ This work was supported
 in part by the National Science Foundation under grant ANI-0073725, and by the
European Commission under contract no. MIRG-CT-2004-006666.}}
\author{Suat G\"um\"u\c{s}soy$^{\dag}$\thanks{$^{\dag}$ was with Dept. of Electrical and Computer Eng., The Ohio State
University, Columbus, OH 43210, U.S.A.; current affiliation: MIKES
Inc., Akyurt, ANKARA TR-06750, Turkey, {\sl
suat.gumussoy@mikes.com.tr}} ~~ and~~  Hitay
\"Ozbay$^{\ddag}$\thanks{$^{\ddag}$ Dept. of Electrical and
Electronics Eng., Bilkent University, Bilkent, Ankara TR-06800,
Turkey, on leave from Dept. of Electrical and Computer Eng., Ohio
State University, Columbus,
OH 43210, {\sl hitay@bilkent.edu.tr, ozbay.1@osu.edu}}}%
\date{}

\begin{document}
\maketitle
\begin{abstract}
A state space based design method is given to find strongly
stabilizing controllers for multi-input-multi-output plants
(MIMO). A sufficient condition is derived for the existence of
suboptimal stable $\Hi$ controller in terms of linear matrix
inequalities (LMI) and the controller order is twice that of the
plant. A new parameterization of strongly stabilizing controllers
is determined using linear fractional transformations (LFT).
\end{abstract}

\section{Introduction}
\setcounter{equation}{0}

Strong stabilization problem is known as the design of a stable
feedback controller which stabilizes the given plant. For
practical reasons, a stable controller is desired
\cite{V-CSS-85,ZO-ACC-97}. In this paper, we derive a simple and
effective design method to find stable $\Hi$ controllers for MIMO
systems.

A stable controller can be designed if and only if the plant
satisfies the parity interlacing property (PIP) \cite{YBL-AUT-74}
i.e., the plant has even number of poles between any pair of its
zeros on the extended positive real axis. There are several design
procedure for strongly stabilizing controllers, [4--16].

The result in this paper is the generalization of the work in
\cite{ZO-AUT-00} using LMIs. The procedure is quite simple,
efficient and easy to solve by using the LMI Toolbox of MATLAB
\cite{GNLC-LMI-95}. In the next section, it is shown that if a
certain LMI has a feasible solution, then it is possible to obtain
a stable $\Hi$ controller whose order is twice the order of the
plant. Moreover, a parameterization of strongly stabilizing
controllers can be given in terms of LFT.

The paper is organized as follows. The main results are given in
Section 2. Stable $\Hi$ controller design procedure is proposed in
Section 3. Numerical examples and comparison with other methods
can be found in Section 4 and concluding remarks are made in the
last section.

\textbf{Notation}\\
The notation is fairly standard. A state space realization of a
transfer function, $G(s)=C(sI-A)^{-1}B+D$, is shown by $
G(s)=\left[\begin{array}{c|c}
  A & B \\
  \hline
  C & D \\
\end{array}\right]$ and the linear fractional transformation of $G$ by $K$ is
denoted by $\mathcal{F}_l(G,K)$ which is equivalent to
$G_{11}+G_{12}K(I-G_{22}K)^{-1}G_{21}$ where $G$ is partitioned as
$G=\left[\begin{array}{cc}
  G_{11} & G_{12} \\
  G_{21} & G_{22} \\
\end{array}\right]$. As a shorthand  notation for LMI expressions, we will
define $\Gamma(A,B):=B^T A^T+A B$ where $A$, $B$ are matrices with
compatible dimensions.

\section{Strong stabilization of MIMO systems}
Consider the standard feedback system with generalized plant, $G$,
which has state space realization, \be \label{eq:plant}
G(s)=\left[\begin{array}{c|cc}
  A & B_{1} & B \\
  \hline
  C_{1} & D_{11} & D_{12} \\
  C & D_{21} & 0 \\
\end{array}\right]
\ee where $A\in\mathcal{R}^{n\times n}$,
$D_{12}\in\mathcal{R}^{p_1\times m_2}$,
$D_{21}\in\mathcal{R}^{p_2\times m_1}$ and other matrices are
compatible with each other. We suppose the plant satisfies the
standard assumptions,
\begin{enumerate}
    \item[A.1] $(A,B)$ is stabilizable and $(C,A)$ is detectable,
    \item[A.2] $\left[\begin{array}{cc}
  A-\lambda I & B \\
  C_1 & D_{12} \\
\end{array}\right]$ has full column rank for all $Re\{{\lambda}\}\geq
0$,
    \item[A.3] $\left[\begin{array}{cc}
  A-\lambda I & B_1 \\
  C & D_{21} \\
\end{array}\right]$ has full row rank for all $Re\{{\lambda}\}\geq
0$,
    \item[A.4] A has no eigenvalues on the imaginary axis.
\end{enumerate}
Let the controller has state space realization,
$K_G(s)=\left[\begin{array}{c|c}
  A_K & B_K \\
  \hline
  C_K & 0 \\
\end{array}\right]$ where $A_K\in\mathcal{R}^{n\times n}$, $B_K\in\mathcal{R}^{n\times p_2}$ and
$C_K\in\mathcal{R}^{m_2\times n}$. Define the matrix
$X\in\mathcal{R}^{n\times n}$, $X=X^T>0$ as the stabilizing
solution of \be \label{eq:riccati} A^T X+X A-XBB^TX=0 \ee (i.e.,
$A-BB^TX$ is stable) and the ``A-matrix" of the closed loop system
as $A_{CL}=\left[\begin{array}{cc}
A & B C_K \\
B_KC & A_K\\
\end{array}\right]$. Note that since $(A,B)$ is stabilizable, $X$  is unique and $A_X:=(A-BB^TX)$ is stable.
Also, the closed loop stability is equivalent to whether $A_{CL}$
is stable or not.

\begin{lemma} \label{lemma:LMI} Assume that the plant (\ref{eq:plant}) satisfies the
assumptions $A.1-A.4$. There exists a stable stabilizing
controller, $K_G\in\mathcal{RH}^\infty$ if there exists
$X_K\in\mathcal{R}^{n\times n}$, $X_K=X_K^T>0$ and
$Z\in\mathcal{R}^{n\times p_2}$ for some $\gamma_K>0$ satisfying
the LMIs, \be \label{eq:LMI1} \Gamma(X_K,A)+\Gamma(Z,C)<0, \ee \be
\label{eq:LMI2} \left[\begin{array}{ccc}
  \Gamma(X_K,A_X)+\Gamma(Z,C) & -Z & -X B \\
  -Z^T & -\gamma_K I & 0 \\
  -B^T X & 0 & -\gamma_K I \\
\end{array}\right] <0,
\ee
where $X$ is the stabilizing solution of (\ref{eq:riccati}) and
$A_X$ is as defined previously. Moreover, under the above
condition, a stable controller can be given as
$K_G(s)=\left[\begin{array}{c|c}
  A_X+X_K^{-1}ZC & -X_K^{-1}Z \\
  \hline
  -B^TX & 0 \\
\end{array}\right]$ and this controller satisfies $\|K_G\|_\infty<\gamma_K$.
\end{lemma}
\begin{proof} By using similarity transformation, one
can show that $A_{CL}$ is stable if and only if $A_X$ and
$A_Z:=A+X_K^{-1}ZC$ is stable. Since $X$ is a stabilizing
solution, $A_X$ is stable. If we rewrite the LMI (\ref{eq:LMI1})
as \bd (A+X_K^{-1}ZC)^TX_K+X_K(A+X_K^{-1}ZC)<0,\ed it can be seen
that $A_Z$ is stable since $X_K>0$. The second LMI (\ref{eq:LMI2})
comes from KYP lemma and guarantees that
$\|K_G\|_\infty<\gamma_K$.
\end{proof}
\noindent \textbf{Remark 1} If the design only requires the
stability of closed loop system, it is enough to satisfy the LMI
(\ref{eq:LMI1}), $(1,1)$ block of (\ref{eq:LMI2}), i.e., \be
\label{eq:LMI(1,1)} A_X^T X_K+X_K A_X +C^TZ^T+Z C<0 \ee and the
controller has same
structure as above.\\
\noindent \textbf{Remark 2} The Lemma (\ref{lemma:LMI}) is
generalization of Theorem 2.1 in \cite{ZO-AUT-00}. If the
algebraic riccati equation (ARE) $(7)$ in \cite{ZO-AUT-00} has a
stabilizing solution, $Y=Y^T\geq0$, then there exists a stable
controller in the form, $\left[\begin{array}{c|c}
  A_X-\gamma_K^2YC^TC & \gamma_K^2YC^T \\
  \hline
  -B^TX & 0 \\
  \end{array}\right]$. This structure is the special case of the
  LMIs (\ref{eq:LMI1}) and (\ref{eq:LMI2}) when $X_K=(\gamma_K
Y)^{-1}$ and $Z=-\gamma_K C^T$. Note that our formulation does not
assume special structure on $Z$. Also in \cite{ZO-AUT-00}, the
stability of $A_Z$ is guaranteed by the same riccati equation, we
satisfy the stability condition of $A_Z$ with another LMI
(\ref{eq:LMI1}) which is less restrictive. Therefore, the Lemma
(\ref{lemma:LMI}) is less conservative as will be demonstrated in
examples. \\

\begin{corollary}
Assume that the sufficient condition (\ref{eq:LMI1}) and
(\ref{eq:LMI(1,1)}) holds. Then all controllers in the set \bd
\mathrm{K}_{G,ss}:=\{K=\mathcal{F}_l(K_{G,ss}^0,Q):
Q\in\mathcal{RH}^\infty, \|Q\|_\infty<\gamma_Q\} \ed are strongly
stabilizing where \be K_{G,ss}^0(s)=\left[\begin{array}{c|cc}
  A_X+X_K^{-1}ZC & -X_K^{-1}Z & B \\
  \hline
  -B^TX & 0 & I \\
  -C & I & 0 \\
\end{array}\right].
\ee and
$\gamma_Q=\left(\|C(sI-(A_X+X_K^{-1}ZC))^{-1}B\|_\infty\right)^{-1}$.
\end{corollary}
\begin{proof} The result is direct consequence of
parameterization of all stabilizing controllers \cite{ZDG-ROC-96}.
\end{proof}

\section{Stable $\Hi$ controller design for MIMO systems}
The standard $\Hi$ problem is to find a stabilizing controller $K$
such that $\|\mathcal{F}_l(P,K)\|_\infty\leq\gamma$ where
$\gamma>0$ is the closed loop performance level and $P$ is the
generalized plant. It is well known that if two AREs have unique
positive semidefinite solutions and the spectral radius condition
is satisfied, then standard $\Hi$ problem is solvable. All
suboptimal $\Hi$ controllers can be parameterized as
$K=\mathcal{F}_l(M_\infty,Q)$ where the central controller is in
the form \bd M_\infty(s)=\left[\begin{array}{c|cc}
  A_c & B_{c1} & B_{c2} \\
  \hline
  C_{c1} & D_{c11} & D_{c12} \\
  C_{c2} & D_{c21} & 0 \\
\end{array}\right]
\ed and $Q$ is free parameter satisfying $Q\in\mathcal{RH}^\infty$
and $\|Q\|_\infty\leq\gamma$. For derivation and calculation of
$M_\infty$, see \cite{DGKF-TAC-89,ZDG-ROC-96}.

If we consider $M_\infty$ as plant and $\gamma=\gamma_K$, by using
Lemma (\ref{lemma:LMI}), we can find a strictly proper stable
$K_{M_\infty}$ stabilizing $M_\infty$ and resulting stable $\Hi$
controller, $C_{\gamma}= \mathcal{F}_l(M_\infty,K_{M_\infty})$
where $\|K_{M_\infty}\|_\infty<\gamma_K$. If sufficient conditions
(\ref{eq:LMI1}) and (\ref{eq:LMI2}) are satisfied, then
$K_{M_\infty}$ can be written as, \bd
K_{M_\infty}(s)=\left[\begin{array}{c|c}
  A_c-B_{c2}B_{c2}^TX_c+X_{Kc}^{-1}Z_cC_{c2} & -X_{Kc}^{-1}Z_c \\
  \hline
  -B_{c2}^TX_c & 0 \\
\end{array}\right]
\ed and by similarity transformation, we can obtain the state
space realization of $C_\gamma$ as,
\bd
C_\gamma(s)=\left[\begin{array}{c|c}
A_{C_\gamma} & B_{C_\gamma}\\
\hline
C_{C_\gamma} & D_{c11}\\
\end{array}\right]
\ed where $X_c$ is the stabilizing solution of \bd
A_c^TX_c+X_cA_c-X_cB_{c2}B_{c2}^TX_c=0\ed as in (\ref{eq:riccati})
and $X_{Kc}$, $Z_c$ are the solution of (\ref{eq:LMI1}),
(\ref{eq:LMI2}) respectively and the matrices, \bea \nonumber
A_{C_\gamma}&=&\left[\begin{array}{cc}
  A_c-B_{c2}B_{c2}^TX_c & -B_{c2}B_{c2}^TX_c \\
  0 & A_c+X_{Kc}^{-1}Z_cC_{c2} \\
  \end{array}\right] \\
\nonumber B_{C_\gamma}&=&\left[\begin{array}{c}
B_{c1} \\
 -B_{c1}-X_{Kc}^{-1}Z_cD_{c21} \\
\end{array}\right] \\
\nonumber C_{C_\gamma}&=&\left[\begin{array}{cc}
C_{c1}-D_{c12}B_{c2}^T X_c & -D_{c12}B_{c2}^T X_c\\
\end{array}\right]
\eea
Note that $C_\gamma$ is stable stabilizing controller such that
$\|\mathcal{F}_l(P,C_\gamma)\|_\infty<\gamma$.
\section{Numerical examples and comparisons}
\subsection{Strong stabilization} The numerical example is chosen
from \cite{ZO-AUT-00}. In order to see the performance of our
method, we calculated the minimum $\gamma_K$ satisfying the
sufficient conditions in Lemma (\ref{lemma:LMI}) for the following
plants:
\bea
\nonumber G_1(s)&=&\left[\begin{array}{c}
\frac{(s+5)(s-1)(s-5)}{(s+2+j)(s+2-j)(s-\alpha)(s-20)} \\
\\
\frac{(s+1)(s-1)(s-5)}{(s+2+j)(s+2-j)(s-\alpha)(s-20)}
\end{array}\right] \\
\nonumber G_2(s)&=&\left[\begin{array}{c}
\frac{(s+1)(s-2-j\alpha)(s-2+j\alpha)}{(s+2+j)(s+2-j)(s-1)(s-5)} \\
\\
\frac{(s+5)(s-2-j\alpha)(s-2+j\alpha)}{(s+2+j)(s+2-j)(s-1)(s-5)}
\end{array}\right]
\eea
For various $\alpha$ values, the minimum $\gamma_K$
is found. Figure \ref{fig:G1comp} and \ref{fig:G2comp} illustrates
the conservatism of \cite{ZO-AUT-00} mentioned in Remark 2 (where
$\rho_{min}$ is the minimum value of the free parameter $\gamma_K$
corresponding to the method of \cite{ZO-AUT-00}).

\begin{figure}[h!]
\begin{center}
\begin{minipage}[b]{0.5\textwidth}
\centerline{
\includegraphics[width=3in,height=2.5in]{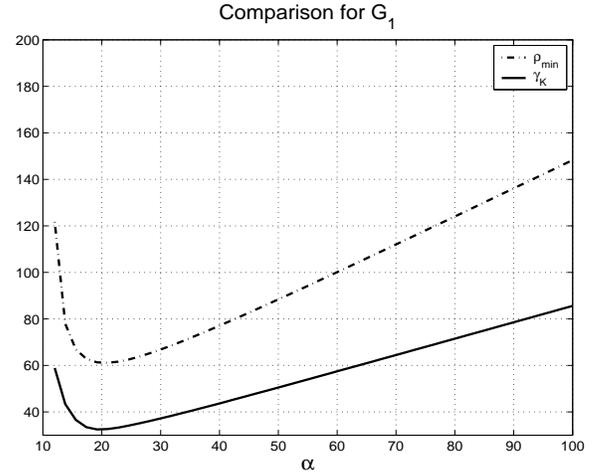}}
\caption{Comparison for plant $G_1$} \label{fig:G1comp}
\end{minipage}%
\end{center}
\end{figure}

\begin{figure}[h!]
\begin{center}
\begin{minipage}[b]{0.5\textwidth}
\centerline{
\includegraphics[width=3in,height=2.5in]{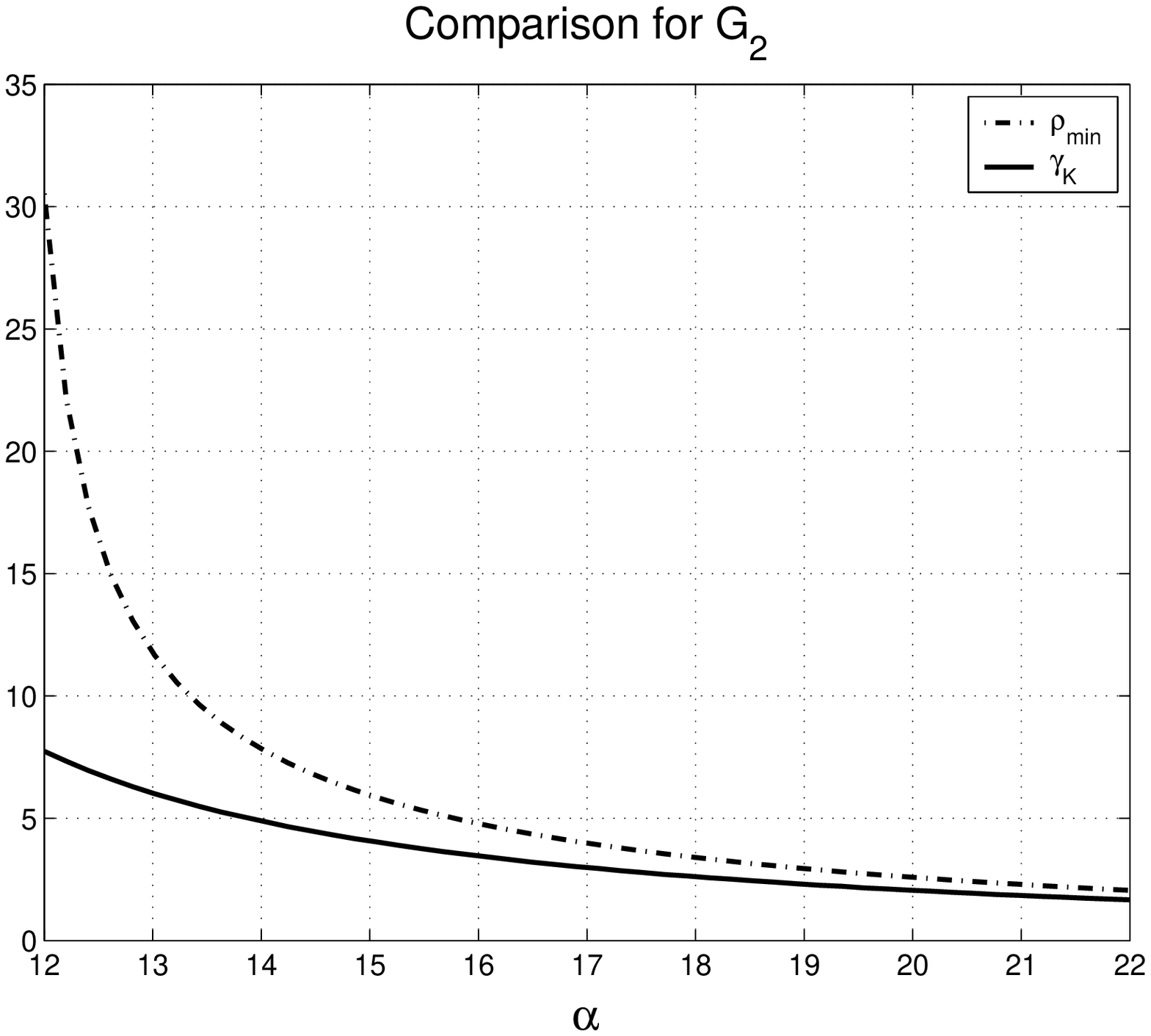}}
\caption{Comparison for plant $G_2$} \label{fig:G2comp}
\end{minipage}
\end{center}
\end{figure}

\subsection{Stable $\Hi$ controllers}

We applied our method to stable $\Hi$ controller design. As a
common benchmark example, the following system is taken from
\cite{LS-SCL-02}: \be P=\left[\begin{array}{c|cc}
  A & B_1 & B_2 \\
  \hline
  C_1 & D_{11} & D_{12} \\
  C_2 & D_{21} & 0 \\
\end{array}\right]
\ee where
\bea
\nonumber A&=& \left[\begin{array}{cc}
    -2 & 1.7321 \\
    1.7321 & 0 \\
  \end{array}\right] \\
\nonumber \left[\begin{array}{c|c} B_1 & B_2 \end{array}\right]&=&\left[\begin{array}{cc|c}
    0.1 & -0.1 & 1\\
    -0.5 & 0.5 & 0\\
  \end{array}\right] \\
\nonumber \left[\begin{array}{c}
C_1 \\
\hline
C_2 \end{array}\right]&=& \left[\begin{array}{cc}
    0.2 & -1 \\
    0 & 0 \\
    \hline
    10 & 11.5470 \\
  \end{array}\right] \\
\nonumber  \left[\begin{array}{c|c}
D_{11} & D_{12}\\
\hline
D_{21} & 0 \\
\end{array}\right]&=&\left[\begin{array}{cc|c}
    0 & 0 & 0 \\
    0 & 0 & 1 \\
    \hline
    0.7071 & 0.7071 & 0 \\
  \end{array}\right]
\eea

The optimal $\gamma$ value for standard $\Hi$ problem is
$\gamma_{opt}=1.2929$. Using the synthesis in \cite{LS-SCL-02}, a
stable $\Hi$ controller is found at $\gamma_{min}=1.36994$. When
our method applied, we reached stable $\Hi$ controller for
$\gamma_{K,min}=1.36957$. Although it seems slight improvement,
our method is much more simpler with help of LMI problem
formulation. Apart from standard problem solution (finding
$M_\infty$), the algorithm in \cite{LS-SCL-02} finds the stable
$\Hi$ controller by solving an additional $\Hi$ problem.

Another common benchmark example (see \cite{CZ-TAC-01} and its
references) is to find a stable $\Hi$ controller for the
generalized plant described by (\ref{BenchExy}).

\begin{figure*}[!t]
\bea \nonumber
z_1&=&\frac{0.03s^7+0.008s^6+0.19s^5+0.037s^4+0.36s^3+0.05s^2+0.18s+0.015}
{s^8+0.161s^7+6s^6+0.582s^5+9.984s^4+0.407s^3+3.9822s^2}(w_1+u), \\
 z_2&=&\beta u,  \nonumber \\
\label{BenchExy}
y&=&w_2+\frac{0.0064s^5+0.0024s^4+0.071s^3+s^2+0.1045s+1}
{s^8+0.161s^7+6s^6+0.582s^5+9.984s^4+0.407s^3+3.9822s^2}(w_1+u).
\eea
\hrulefill
\vspace*{4pt}
\end{figure*}

\begin{table}[b]
\caption{Stable $\Hi$ controller design for (\ref{BenchExy})}
\label{comparison_table}
\begin{center}
\begin{tabular}{|c|c|c|c|c|}
  \hline
  \hline
   & & Gumussoy-Ozbay(GO) & \cite{ZO-TAC-99} & \cite{CZ-TAC-01} \\
  \hline
  $\beta$ & $\gamma_{opt}$ & $\gamma_{GO}$ & $\gamma_{ZO}$ & $\gamma_{CZ}$ \\
  \hline
  $0.1$ & $0.232$ & $0.241$ & $0.245$ & $0.237$ \\
  $0.01$ & $0.142$ & $0.176$ & $0.178$ & $0.151$ \\
  $0.001$ & $0.122$ & $0.170$ & $0.170$ & $0.132$ \\
  \hline
  \hline
\end{tabular}
\end{center}
\end{table}

In \cite{ZO-TAC-99}, it is noted that for this problem, the
sufficient condition in \cite{IOS-TAC-93} is not satisfied for
even large values of $\gamma$ and the method is not applicable. As
we can see from Table \ref{comparison_table}, the performance of
our method is better than the method in \cite{ZO-TAC-99} except
the last case. For all cases, the result of \cite{CZ-TAC-01} is
superior from all other methods. However, the controller order in
\cite{CZ-TAC-01} is $24$ which is greater than our controller
order, $16$.  To address this problem, in \cite{CZ-TAC-01}  a
controller order reduction is performed, that results in lower
order (e.g. 10th order for the case $\beta =0.1$) stable
controllers without significant loss of performance. Furthermore,
the method in \cite{CZ-TAC-01} involves solution of an additional
$\Hi$ problem which is complicated compared to our simple LMI
formulation. If the algorithm in \cite{CZ-TAC-01} fails, selection
of a new parameter $Q$ is suggested which is an ad-hoc procedure.
Although the performance of the controller suggested in the
present paper is slightly worse, it is numerically stable and
easily formulated.

The following example is taken from \cite{CC-TAC-01}. Design a
stable $\Hi$ controller  for the  plant
\bd
P(s)=\frac{s^2+0.1s+0.1}{(s-0.1)(s-1)(s^2+2s+3)}.
\ed 
For the mixed sensitivity minimization problem the weights are
taken to be as in \cite{CC-TAC-01}. A comparison of the methods
\cite{ZO-TAC-99,CC-TAC-01,CWL-CDC-03} and our method can be seen
in Table~\ref{comparison_tableII}. There is a compromise between
the methods. The performance of the method in \cite{CC-TAC-01} is
worse than our method, but the order of  our controller has twice
order of the controller in \cite{CC-TAC-01}. Although the method
in \cite{CWL-CDC-03} gives better results than our method, the
order of the controller in \cite{CWL-CDC-03} is considerably
higher than our controller order. However, this example clearly
shows that our method is superior than \cite{ZO-TAC-99}.

\begin{table}[b]
\caption{Stable $\Hi$ controller design for Example in
\cite{CC-TAC-01}} \label{comparison_tableII}
\begin{center}
\begin{tabular}{|c|c|c|c|c|}
  \hline
  \hline
   & Gumussoy-Ozbay & \cite{ZO-TAC-99} & \cite{CC-TAC-01} & \cite{CWL-CDC-03}, Thm $7$ \\
  \hline
  $\gamma_{min}$ & $32.557$ & $37.551$ & $43.167$ & $21.787$ \\
  Order & $2n$ & $2n$ & $n$ & $3n$ \\
  \hline
  \hline
\end{tabular}
\end{center}
\end{table}

As a remark, the method also gives very good results for
single-input-single-output systems. The following SISO example is
taken from \cite{ZO-AUT-00}:
\bd
P(s)=\frac{(s+5)(s-1)(s-5)}{(s+2+j)(s+2-j)(s-20)(s-30)},
\ed
\bd
W_1(s)=\frac{1}{s+1},
\ed
\bd
W_2(s)=0.2,
\ed
the optimal $\Hi$
problem is defined as \bd \gamma_{opt}=\inf_{K stabilizing
P}\left\|\left[\begin{array}{c}
  W_1(1+PK)^{-1} \\
  W_2K(1+PK)^{-1} \\
\end{array}\right]\right\|_\infty
\ed and the optimal performance for the given data is
$\gamma_{opt}=34.24$. A stable $\Hi$ controller can be found for
$\gamma=42.51$ using the method of \cite{ZO-AUT-00}, whereas our
method, which can be seen as a generalization of \cite{ZO-AUT-00},
gives a stable controller with $\gamma=35.29$. In numerical
simulations, we observed that when $\gamma$ approaches to the
minimum value satisfying sufficient conditions, the solutions of
algebraic riccati equations of \cite{ZO-AUT-00} become numerically
ill-posed. However, the LMI based solution proposed here does not
have such a problem. Same example is considered in
\cite{CZ-TAC-03} and stable $\Hi$ controller found for
$\gamma=34.44$. The method in \cite{CZ-TAC-03} is a two-stage
algorithm with  combination of genetic algorithm and quasi-Newton
algorithm and gives slightly better performance than our method.
The method finds stable $\Hi$ controllers with a selection of
low-order controller for free parameter $Q$. Since the example
considered in the paper is for SISO case, it may be difficult to
achieve good performance with low-order controller for MIMO case.
Due to nonlinear optimization problem structure, the solution of
the method may converge to local minima and in general, genetic
algorithms give solution for longer time.

\section{Concluding Remarks}
In this paper, sufficient conditions for strong stabilization of
MIMO systems are obtained and applied to stable $\Hi$ controller
design. Our conditions are based on linear matrix inequalities
which can be easily solved by the LMI Toolbox of MATLAB. The
method is very efficient from numerical point of view as
demonstrated with examples. The benchmark examples show that the
proposed method is a significant improvement over the existing
techniques available in the literature. The exceptions to this
claim are the methods of \cite{CZ-TAC-01,CWL-CDC-03,CZ-TAC-03}. In
\cite{CZ-TAC-01}, the controller design is based on ad-hoc search
method, and both \cite{CC-TAC-01} and \cite{CWL-CDC-03} result in
higher order controllers than the one designed by our method. In
\cite{CZ-TAC-03}, selection of low-order controller for $Q$ gives
good results for SISO structure of $Q$. However in MIMO structure,
$Q$ may not result in good performance.

\end{document}